\documentclass{llncs}

\usepackage[english]{babel}
\usepackage[T1]{fontenc}
\usepackage{amsmath, amssymb, amsbsy}
\usepackage{mathdots}
\usepackage[linesnumbered,ruled,vlined,titlenumbered]{algorithm2e}
\usepackage{color}
\usepackage{cite}

\usepackage{enumitem}
\usepackage{colortbl}

\newtheorem{prob}{Problem}

\newcommand\qbin[3]{\left[\begin{matrix} #1 \\ #2 \end{matrix} \right]_{#3}}

\def\ve#1{{\mathchoice{\mbox{\boldmath$\displaystyle #1$}}%
              {\mbox{\boldmath$\textstyle #1$}}%
              {\mbox{\boldmath$\scriptstyle #1$}}%
              {\mbox{\boldmath$\scriptscriptstyle #1$}}}}

\definecolor{newcolor}{rgb}{0, 0.5, 0}
\newcommand{\new}[1]{#1}

\renewcommand{\vec}[1]{\ve{#1}}

\newcommand{\Fq}{\ensuremath{\mathbb{F}_q}}
\newcommand{\Fqm}{\ensuremath{\mathbb{F}_{q^m}}}

\DeclareMathOperator{\rank}{rk}

\newcommand{\Code}{\mathcal{C}}

\newcommand{\rspace}[2]{\mathcal{R}_{#2} \begin{pmatrix}#1\end{pmatrix}}

\renewcommand{\r}{\ve{r}}

\renewcommand{\a}{\ve{a}}

\newcommand{\e}{\ve{e}}

\newcommand{\g}{\ve{g}}
\newcommand{\m}{\ve{m}}
\renewcommand{\c}{\ve{c}}
\renewcommand{\e}{\ve{e}}
\renewcommand{\m}{\ve{m}}

\newcommand{\GGab}{\ve{G}_\mathrm{Gab}}

\newcommand{\A}{\ve{A}}
\newcommand{\B}{\ve{B}}

\newcommand{\ac}{\ve{a}_{\text{C}}}

\newcommand{\ar}{\ve{a}_{\text{R}}}
\renewcommand{\ae}{\ve{a}_{\text{E}}}

\newcommand{\Bc}{\ve{B}_{\text{C}}}
\newcommand{\Br}{\ve{B}_{\text{R}}}
\newcommand{\Be}{\ve{B}_{\text{E}}}
\newcommand{\BcHat}{\hat{\ve{B}}_{\text{C}}}

\newcommand{\dimI}{\epsilon}

\newcommand{\CGab}{Gab_k(\g)}

\newcommand{\Uspace}{\mathcal{U}}
\newcommand{\Vspace}{\mathcal{V}}
\newcommand{\Fql}{\mathbb{F}_{q^{\ell}}}
\newcommand{\jstar}{j^{*}}

\newcommand{\dec}[1]{\ensuremath{\textsf{Dec}(#1)}}

\newcommand{\assignDet}{\ensuremath{\leftarrow}}
\newcommand{\assignRand}{\ensuremath{\xleftarrow{\$}}}
\newcommand{\Grassm}[2]{\mathcal{G}(#1,#2)}

\begin{document}

\title{Randomized Decoding of Gabidulin Codes Beyond the Unique Decoding Radius}
\titlerunning{Randomized Decoding of Gabidulin Codes}
 
\author{Julian Renner\inst{1} \and
Thomas Jerkovits\inst{2} \and
Hannes Bartz\inst{2} \and
Sven Puchinger\inst{3} \and \newline
Pierre Loidreau\inst{4} \and
Antonia Wachter-Zeh\inst{1}}

\authorrunning{Renner, Jerkovits, Bartz, Puchinger, Loidreau, Wachter-Zeh}

\institute{
  Technical University of Munich (TUM), Munich, Germany\\
  \email{\{julian.renner, antonia.wachter-zeh\}@tum.de}\thanks{The work of J. Renner and A. Wachter-Zeh was supported by the European Research Council (ERC) under the European Union's Horizon 2020 research and innovation programme (grant agreement No 801434).} \and
German Aerospace Center (DLR), Oberpfaffenhofen-Wessling, Germany \\\email{\{hannes.bartz, thomas.jerkovits\}@dlr.de} \and
  Technical University of Denmark (DTU), Lyngby, Denmark \\\email{svepu@dtu.dk}\thanks{Sven Puchinger has received funding from the European Union's
Horizon 2020 research and innovation program under the Marie Sklodowska-Curie grant
agreement no.~713683 (COFUNDfellowsDTU).} \and
Univ Rennes, DGA MI, CNRS, IRMAR - UMR 6625, F-35000 Rennes, France \\\email{pierre.loidreau@univ-rennes1.fr }
}

\maketitle       

\begin{abstract}
We address the problem of decoding Gabidulin codes beyond their unique error-correction radius. 
The complexity of this problem is of importance to assess the security of some rank-metric code-based cryptosystems. We propose an approach that introduces row or column erasures to decrease the 
rank of the error in order to use any proper polynomial-time Gabidulin code error-erasure decoding algorithm.  
\new{The expected work factor of this new randomized decoding approach is a polynomial term times $q^{m(n-k)-w(n+m)+w^2+\min\{2\xi(\frac{n+k}{2}-\xi),wk\} }$, where $n$ is the code length, $q$ the size of the base field, $m$ the extension degree of the field, $k$ the code dimension, $w$ the number of errors, and $\xi := w-\tfrac{n-k}{2}$.
It improves upon generic rank-metric decoders by an exponential factor.}
\end{abstract}

\section{Introduction}
\new{Code-based cryptography relies on the hardness of certain coding-theoretic problems, e.g., decoding a random code up to its unique decoding radius or, as considered in this paper, decoding more errors than the unique decoding radius and beyond the capabilities of all known polynomial-time decoding algorithms. Rank-metric schemes that rely on the latter problem have the promising potential to achieve key sizes that are \emph{linear} in the security parameter and are for instance the (modified) Faure--Loidreau system~\cite{faure2006new,wachter2018repairing} or the RAMESSES system~\cite{lavauzelle2019}.}

In the Hamming metric as well as in the rank metric, it is well-known that the problem of decoding beyond the unique decoding radius, in particular \emph{Maximum-Likelihood} (ML) decoding, is a difficult problem concerning the complexity. In Hamming metric, many works have analyzed how hard it actually is, cf.~\cite{BMT78,S93B}, and it was finally shown for general linear codes that ML decoding is NP-hard by Vardy in \cite{V97}. 
For the rank metric, some complexity results were obtained more recently in \cite{GZ15}, 
emphasizing the difficulty of ML decoding. This potential hardness was also corroborated by the existing practical complexities of the generic rank metric decoding algorithms \cite{gaborit2016decoding}.

For specific well-known families of codes such as \emph{Reed--Solomon} (RS) codes in the Hamming metric, (ML or list) decoding can be done efficiently up to a certain radius. 
Given a received word, an ML decoder returns \emph{the} (or one if there is more than one) \emph{closest codeword} to the received word whereas a list decoder returns \emph{all codewords} up to a fixed radius. 
The existence of an efficient list decoder up to a certain radius therefore implies an efficient ML decoder up to the same radius. Vice versa, this is however not necessarily true, but we cannot apply a list decoder to solve the ML decoding problem efficiently.

In particular, for an RS code of length $n$ and dimension $k$, the following is known, depending on the Hamming weight $w$ of the error:\\[-3.5ex]
\begin{itemize}
	\item If $w \leq \left\lfloor\frac{n-k}{2}\right \rfloor$, the (ML and list) decoding result is unique and can be found in quasi-linear time,
	\item If $w < n-\sqrt{n(k-1)}$, i.e., the weight of the error is less than the Johnson bound, list decoding and therefore also ML decoding can be done efficiently by Guruswami--Sudan's list decoding algorithm \cite{GS99},
	\item The renewed interest in  RS codes after the design of the Guruswami--Sudan list decoder \cite{GS99} motivated new studies of the theoretical complexity of ML and list decoding of RS codes. 
	In \cite{GV05} it was shown that ML decoding of 
	RS codes is indeed NP-hard when $w \geq d-2$, even with some pre-processing.
	\item Between the Johnson radius and $d-2$, it has been shown in \cite{BenSasson2010Subspace} that the number of codewords in radius $w$ around the received word might become a number that grows super-polynomially in $n$ which makes list decoding of RS codes a hard problem. 
\end{itemize}

\emph{Gabidulin codes} \cite{Delsarte_1978,Gabidulin_TheoryOfCodes_1985,Roth_RankCodes_1991} can be seen as the rank-metric analog of RS codes. ML decoding of Gabidulin codes is in the focus of this paper which is much less investigated than for RS codes (see the following discussion). However, both problems (ML decoding of RS and Gabidulin codes) are of cryptographic interest. The security of the Augot--Finiasz public-key cryptosystem from \cite{AF03} relied on the hardness of ML decoding of RS codes but was broken by a structural attack. More recently, some public-key cryptosystems based their security partly upon the difficulty of solving the problem \textsf{Dec-Gab} (\textsf{Decisional-Gabidulin} defined in the following) and \textsf{Search-Gab} (\textsf{Search-Gabidulin}), i.e., decoding Gabidulin codes beyond the unique decoding radius or derived instances of this problem \cite{faure2006new,wachter2018repairing, lavauzelle2019}. 

\textsf{Dec-Gab} has not been well investigated so far. Therefore, we are interested in designing efficient algorithms to solve \textsf{Dec-Gab} which in turn assesses the security of several public-key cryptosystems.
We deal with analyzing the problem of decoding Gabidulin codes beyond the unique radius where a Gabidulin code of length~$n$ and dimension $k$ is denoted by $Gab_k(\g)$ and $\g = (g_0,g_1,\dots,g_{n-1})$ denotes the vector of linearly independent code locators.
\begin{prob}[\textsf{Dec-Gab}]\hfill
\begin{itemize}
  \item Instance: $Gab_k(\g) \subset \Fqm^n$, $\r \in \Fqm^n$ and an integer $w >0$.
  \item Query: Is there a codeword $\c \in Gab_k(\g)$, such that $\rank(\r - \c) \le w$?   
\end{itemize}
\end{prob}

It is trivial that \textsf{Dec-Gab}($Gab_k(\g)$, $\r$,  $w$) can be solved in deterministic polynomial time whenever:
\begin{itemize}
  \item $w \le\left\lfloor\frac{n-k}{2}\right \rfloor$, with applying a deterministic polynomial-time decoding algorithm for Gabidulin codes to $\r$.
  \item $w \ge n-k$: In this case the answer is always \texttt{yes} since this just tantamounts to finding a solution to an overdetermined full rank linear system (Gabidulin codes are {\em Maximum Rank Distance} codes). 
\end{itemize}

However, between $\left\lfloor\frac{n-k}{2}\right \rfloor$ and $n-k$, the situation for \textsf{Dec-Gab} is less clear than for RS codes (which was analyzed above). 

For instance, concerning  RS codes, the results from \cite{GV05} and \cite{BenSasson2010Subspace} state that there is a point in the interval $[ \left\lfloor\frac{n-k}{2}\right \rfloor, n-k]$ where the situation is not solvable in polynomial-time unless the  polynomial hierarchy collapses. 
For RS codes, we can refine the  interval to $[ n-\sqrt{n(k-1)}, n-k]$, because of the Guruswami-Sudan polynomial-time list decoder up to Johnson bound \cite{GS99}. 

On the contrary, for Gabidulin codes, there is no such a refinement. In \cite{Wachterzeh_BoundsListDecodingRankMetric_IEEE-IT_2013}, it was shown that for \emph{all} Gabidulin codes, the list size grows exponentially in $n$ when $w> n-\sqrt{n(k-1)}$. Further, \cite{RavivWachterzeh_GabidulinBounds_journal} showed that the size of the list is exponential for some Gabidulin codes  as soon as $w =  \left\lfloor\frac{n-k}{2}\right \rfloor + 1$. This result was recently generalized in \cite{ListDec_RankMetric_2019} to other classes of Gabidulin codes (e.g., twisted Gabidulin codes) and, more importantly, it showed that any Gabidulin code of dimension at least two can have an exponentially-growing list size for $w \geq  \left\lfloor\frac{n-k}{2}\right \rfloor + 1$.

To solve the decisional problem \textsf{Dec-Gab} we do not know a better approach than trying to solve the associated \emph{search} problem, which is usually done for all decoding-based problems.

\begin{prob}[\textsf{Search-Gab}]\hfill\label{prob:Search-MLD-Gab}
\begin{itemize}
  \item Instance: $Gab_k(\g) \subset \Fqm^n$, $\r \in \Fqm^n$ and an integer $w >0$. 
  \item Objective: Search for a codeword $\c \in Gab_k(\g)$, such that $\rank(\r - \c) \le w$.
  \end{itemize}
\end{prob}

Since \textsf{Dec-Gab} and \textsf{Search-Gab} form the security core of some rank-metric based cryptosystems, it is necessary to evaluate the effective complexity of 
solving these problems to be able to parameterize the systems in terms of security.

\new{In particular, the problems \textsf{Dec-Gab} and \textsf{Search-Gab} are related to the NIST submission RQC \cite{melchor2019rqc}, the (modified) Faure--Loidreau (FL) cryptosystem \cite{faure2006new,wachter2018repairing}, and RAMESSES \cite{lavauzelle2019}.
	
A part of the security of the newly proposed RAMESSES system \cite{lavauzelle2019} directly relies on the hardness of \textsf{Search-Gab} as solving \textsf{Search-Gab} for the public key directly reveals an alternative private key.

The (modified) FL cryptosystem \cite{faure2006new,wachter2018repairing} is based on the hardness of decoding Gabidulin codes beyond their unique decoding radius. Both, the security of the public key as well as the security of the ciphertext are based on this assumption. The public key can be seen as a corrupted word of an interleaved Gabidulin code whose decoders enabled a structural attack on the original system \cite{Gaborit-KeyRecoveryFaureLoidreau}. In the modified FL system \cite{wachter2018repairing}, only public keys for which all known interleaved decoders fail are chosen, therefore making the structural attack from \cite{Gaborit-KeyRecoveryFaureLoidreau} impossible. As shown in~\cite{JTBH2019}, the component codewords of the public key as well as the ciphertext are a Gabidulin codeword that is corrupted by an error of large weight. Therefore, solving \textsf{Search-Gab} has to be considered when determining the security level of the system.

The NIST submission RQC is based on a similar problem. Here, the ciphertext is also the sum of a Gabidulin codeword and an error of weight larger than the unique decoding radius. 
The error in this setting has a special structure. However, our problem cannot be applied directly to assess the security level of RQC since the error weight is much larger than in the FL and RAMESSES systems and solving \textsf{Search-Gab} for the RQC setting would return a codeword that is close to the error and therefore not the one that was encoded from the plaintext. 
It is not clear how to modify our algorithm to be applicable to RQC since we would have to be able to find exactly the encoded codeword and not just \emph{any} codeword. 
We are not aware of how this can be done but want to emphasize that the underlying problem of RQC is very similar to Problem~\ref{prob:Search-MLD-Gab}.
}

In this paper, we propose a randomized approach to solve \textsf{Search-Gab} and analyze its work factor. \new{The new algorithm consists of repeatedly guessing a subspace that should have a large intersection with the error row and/or column space. Then the guessed space is used as erasures in an Gabidulin error-erasure decoder, e.g.,~\cite{wachter2013decoding,silva2008rank}. The algorithm terminates when the intersection of the guessed space and the error row and/or column space is large enough such that the decoder outputs a codeword that is close enough to the received word $\r$.}

\new{This paper is structured as follows. In Section 2, we introduce the used notation and define Gabidulin codes as well as the channel model. In Section 3, we recall known algorithms to solve \textsf{Search-Gab} and state their work factors. We propose and analyze the new algorithm to solve Problem~\ref{prob:Search-MLD-Gab} in Section 4. Further, numerical examples and simulation results are given in Section 5. Open questions are stated in Section 6.}

\section{Preliminaries}
\subsection{Notation}
Let $q$ be a power of a prime and let
$\Fq$ denote the finite field of order $q$ and $\Fqm$ its extension field of order $q^m$.
This definition includes the important cases for cryptographic applications $q=2$ or $q=2^r$ for a small positive integer $r$.
It is well-known that any element of $\Fq$ can be seen as an element of $\Fqm$ and that $\Fqm$ is an $m$-dimensional vector space over $\Fq$.

We use $\Fq^{m \times n}$ to denote the set of all $m\times n$ matrices over $\Fq$ and $\Fqm^n =\Fqm^{1 \times n}$ for the set of all row vectors of length $n$ over $\Fqm$. Rows and columns of $m\times n$-matrices are indexed by $1,\dots, m$ and $1,\dots, n$, where $A_{i,j}$ is the element in the $i$-th row and $j$-th column of the matrix $\A$. In the following of the paper, we will always consider that $n \le m$. This is the necessary and sufficient condition to design Gabidulin codes.

For a vector $\a \in \Fqm^n$, we define its ($\Fq$-)rank by $\rank(\a) := \dim_{\Fq}\langle a_1,\dots,a_n \rangle_{\Fq}$, where $\langle a_1,\dots,a_n \rangle_{\Fq}$ is the $\Fq$-vector space spanned by the entries $a_i \in \Fqm$ of~$\a$. Note that this rank equals the rank of the matrix representation of $\a$, where the $i$-th entry of $\a$ is column-wise expanded into a vector in $\Fq^m$ w.r.t.\ a basis of $\Fqm$ over $\Fq$.

The Grassmannian $\Grassm{\mathcal{V}}{k}$ of a vector space $\mathcal{V}$ is the set of all $k$-dimensional subspaces of $\mathcal{V}$.

A linear code over $\Fqm$ of length $n$ and dimension $k$ is a $k$-dimensional subspace of $\Fqm^n$ and denoted by $[n,k]_{q^m}$.

\subsection{Gabidulin Codes and Channel Model}\label{sec:channel_model}
Gabidulin codes are a special class of rank-metric codes and can be defined by a generator matrix as follows.
\begin{definition}[Gabidulin Code \cite{Gabidulin_TheoryOfCodes_1985}]
A linear $Gab_k(\g)$ code over $\Fqm$ of length $n \leq m$ 
and dimension $k$ is defined by its $k \times n$ generator matrix 
\begin{equation*}
\GGab = 
\begin{pmatrix}
g_{1} & g_{2} & \dots& g_{n}\\
g_{1}^{q} & g_{2}^{q} & \dots& g_{n}^{q}\\
\vdots &\vdots&\ddots& \vdots\\
g_{1}^{q^{k-1}} & g_{2}^{q^{k-1}} & \dots& g_{n}^{q^{k-1}}\\
\end{pmatrix} \in \Fqm^{k \times n},
\end{equation*}
where $ g_1, g_2, \dots, g_{n} \in \Fqm$ are linearly independent over $\Fq$. 
\end{definition}
The codes are maximum rank distance (MRD) codes, i.e., they attain the maximal possible minimum distance $d=n-k+1$ for a given length $n$ and dimension $k$ \cite{Gabidulin_TheoryOfCodes_1985}.

Let $\r \in \Fqm^{n}$ be a codeword of a Gabidulin code of length $n\leq m$ and dimension $k$ that is corrupted by an error of rank weight $w$, i.e.,
\begin{equation*}
\r = \m \GGab + \e,
\end{equation*}
where $\m \in \Fqm^{k}$, $\GGab \in \Fqm^{k\times n}$ is a generator matrix of an $[n,k]_{q^m}$ Gabidulin code and $\e \in \Fqm^{n}$ with $\rank(\e) = w>\frac{n-k}{2}$. Each error $\e$ of rank weight $w$ can be decomposed into
\begin{equation*}
\e = \a \B,
  \end{equation*}
  where $\a \in \Fqm^{w}$ and $\B \in \Fq^{w\times n}$. The subspace $\langle a_1,\hdots, a_w \rangle_{\Fq}$ is called the column space of the error and the subspace spanned by the rows of $\B$, i.e. $\rspace{\B}{\Fq}$, is called the row space of the error.
  
We define the excess of the error weight $w$ over the unique decoding radius~as
\begin{equation*}
\xi:=w-\frac{n-k}{2}.
\end{equation*}
Note that $2\xi$ is always an integer, but $\xi$ does not necessarily need to be one.

The error $\vec{e}$ can be further decomposed into
\begin{equation}\label{eq:erasureComp}
\e = \ac \Bc + \ar \Br + \ae \Be,
  \end{equation}
  where $\ac \in \Fqm^{\gamma}$, $\Bc \in \Fq^{\gamma \times n}$, $\ar \in \Fqm^{\rho}$, $\Br \in \Fq^{\rho \times n}$, $\ae \in \Fqm^{t}$ and $\Be \in \Fq^{t \times n}$.

  Assuming neither $\ae$ nor $\Be$ are known, the term $\ae \Be$ is called full rank errors. Further, if $\ac$ is unknown but $\Bc$ is known, the product $\ac\Bc$ is called column erasures and assuming $\ar$ is known but $\Br$ is unknown, the vector $\ar\Br$ is called row erasures, see~\cite{silva2008rank,wachter2013decoding}. There exist efficient algorithms for Gabidulin codes~\cite{wachter2013decoding,GabidulinParamonovTretjakov_RankErasures_1991,RichterPlass_DecodingRankCodes_2004,silva2009error} that can correct $\delta := \rho + \gamma$ erasures (sum of row and column erasures) and $t$ errors if
\begin{equation}
2t+\delta  \leq n-k. 
\end{equation}

\section{Solving Problem~\ref{prob:Search-MLD-Gab} Using Known Algorithms}\label{sec:generic}

\subsection{Generic Decoding}

\begin{prob}[\textsf{Search-RSD}]\hfill\label{prob:Search-RSD}
\begin{itemize}
\item Instance: Linear code $\Code \subset \Fqm^n$, $\r \in \Fqm^n$ and an integer $w >0$.
\item Objective: Search for a codeword $\c \in \Code$, such that $\rank(\r - \c) \le w$.
\end{itemize}
\end{prob}

A generic rank syndrome decoding (RSD) algorithm is an algorithm solving Problem~\ref{prob:Search-RSD}. There are potentially many solutions to Problem~\ref{prob:Search-RSD} but for our consideration it is sufficient to find only one of them. 

Given a target vector $\r$ to Problem~\ref{prob:Search-RSD},
the probability  that $\c \in \mathcal{C}$ is such that $\rank(\r - \c) \le w$ is given by
\[
\Pr_{\c \in \mathcal{C}}[\rank(\r - \c) \le w] = \frac{\sum_{i=0}^{w-1}{ \left[\prod_{j=0}^{i-1}{(q^m - q^j)} \right] \qbin{n}{i}{q} } }{q^{mk}}.
\] 

There are two standard approaches for solving Problem~\ref{prob:Search-RSD}. The first method is {\em combinatorial decoding} which consists of enumerating vector spaces. If there is only one solution to the problem, the complexity of decoding an error of rank $w$ in an $[n,k]_{q^m}$ code is equal to
\[
\mathcal{W}_{Comb} = P(n,k) q^{w\lceil (k+1)m/n \rceil - m},
\]
where $P(n,k)$ is a cubic polynomial~\cite{aragon:ISIT18}. In the security evaluations, this polynomial is often neglected and only the exponential term is taken into account. Note that in the case where $m > n$ there might be a better combinatorial bound~\cite{gaborit2016decoding}. Since we do not address this setting, we do not consider this case.

For the evaluation of the post-quantum security, Grover's algorithm has to be taken into account which reduces the complexity of enumeration by a factor of $0.5$ in the exponent. Thus, the estimated complexity is  
 \[
\mathcal{W}_{PQ\_Comb} = P(n,k) q^{0.5(w \lceil (k+1) m/n \rceil - m)}.
\]
Since this is an enumerative approach, the work factors for solving the problem with input $\r$ have to be divided by 
$\mathcal{N} = \max(|\mathcal{C}| \cdot \Pr_{\c \in \mathcal{C}}[\rank(\r - \c) \le w],1)$, corresponding to the estimated number of candidates.

The second approach is {\em algebraic decoding}. It consists of expressing the problem in the form of a multivariate polynomial system and computing a Gr\"obner basis to solve it. A very recent result~\cite{B3GNRT19} estimates rather precisely the cost of the attack and gives generally much better estimations than the combinatorial approach. In case there is a unique solution to the system, then the work factor of the algorithm is

\begin{equation*}
  \mathcal{W}_{Alg} =
\begin{cases}
  O\left( \left[ \frac{((m+n)w)^{w}}{w!} \right]^\mu \right) & \text{if $m{n-k-1 \choose w } \le {n \choose w}$} \\[8pt]
  O\left( \left[ \frac{((m+n)w)^{w+1}}{(w+1)!} \right]^\mu \right) & \text{otherwise, }
  \end{cases}
  \end{equation*}
where $\mu = 2.807$ is the linear algebra constant.
For algebraic decoding, it is neither known how to improve the complexity by using the fact that there are multiple solutions, nor it is known how to speed up the algorithm in the quantum world.

Problem~\ref{prob:Search-MLD-Gab} is a special instance of Problem~\ref{prob:Search-RSD}, where the linear code is a Gabidulin code. In the following, we will show how to reduce the complexity of solving Problem~\ref{prob:Search-MLD-Gab} by using that fact.

\subsection{Key Equation Based Decoding}
In~\cite{Gabidulin_TheoryOfCodes_1985}, a decoding algorithm of Gabidulin codes is presented that is based on solving a linear system of $n-k-w$ equations and $w$ unknowns (called the key equation~\cite[Lemma 4]{Gabidulin_TheoryOfCodes_1985}). If $w > \lfloor \frac{n-k}{2} \rfloor $, there are $w-(n-k-w) = 2\xi$ solutions to this linear system of equations~\cite[Lemma 4]{Wachter2010}, which include all $\c \in Gab_k(\g)$ such that $\rank(\r - \c) \le w$. Brute-force search through all solutions of the key equation solution space for a valid solution to Problem~\ref{prob:Search-MLD-Gab} has a work factor of
\begin{equation*}
\mathcal{W}_{Key} = \frac{n^2 q^{m 2 \xi}}{\mathcal{N}},
\end{equation*}
where checking one solution of the key equation solution space is in $O(n^2)$.

\section{A New Algorithm Solving Problem~\ref{prob:Search-MLD-Gab}}\label{sec:newAlgo}

In the considered problem, $\rank(\e) = w > \frac{n-k}{2}$ and we do not have any knowledge about the row space or the column space of the error, i.e., $\delta=0$ and $t > \frac{n-k}{2}$, meaning that the known decoders are not able to decode $\r$ efficiently.

The idea of the proposed algorithm is to guess parts of the row space and/or the column space of the error and use a basis for the guessed spaces to solve the corresponding error and column/row erasures (see~\eqref{eq:erasureComp}). This approach is a generalization of the algorithm presented in~\cite{JTBH2019}, where only criss-cross erasures are used to decode certain error patterns beyond the unique decoding radius.

The proposed algorithm is given in Algorithm~\ref{alg:guessRS}.
The function $\dec{\r,\vec{a}_R,\vec{B}_C}$ denotes a row/column error-erasure decoder for the Gabidulin code $Gab_k(\g)$ that returns a codeword $\hat{\vec{c}}$ (if $\rank(\vec{r}-\hat{\vec{c}})\leq t+\rho+\gamma$) or $\emptyset$ (decoding failure) and $\delta$ is the total number of guessed dimensions (sum of guessed dimensions of the row space and the column space).
\begin{algorithm}
  \caption{$\textsf{Column-Erasure-Aided Randomized Decoder}$}\label{alg:guessRS}
  \SetKwInOut{Input}{Input}\SetKwInOut{Output}{Output}
  \Input{Received word $\vec{r}\in\Fqm^n$, \\
        Gabidulin error/erasure decoder $\dec{\cdot,\cdot,\cdot}$, \\
        Dimension of guessed row space $\delta$, \\
        Error weight $w$, \\
        Maximum number of iterations $N_{max}$}
  \Output{$\hat{\vec{c}}\in\Fqm^n:\rank(\vec{r}-\hat{\vec{c}})\leq w$ or $\emptyset$ (failure)} 
  \ForEach{$i\in[1,N_{max}]$}{
    $\mathcal{U}\assignRand\Grassm{\Fq^n}{\delta}$ \hfill\tcp{guess $\delta$-dimensional subspace of $\Fq^n$}
    $\vec{B}_C\assignDet$ full-rank matrix whose row space equals $\mathcal{U}$ \\     
    $\hat{\vec{c}}\assignDet\dec{\vec{r},\vec{0},\vec{B}_C}$  \hfill \tcp{error and row erasure decoding} 
    \If{$\hat{\vec{c}}\neq\emptyset$}{
        \If{$\rank(\vec{r}-\hat{\vec{c}})\leq w$}{
          \Return{$\hat{\vec{c}}$}
        }
    }
  }
  \Return{$\emptyset$ (failure)}
\end{algorithm} 

In the following, we derive the work factor of the proposed algorithm. By $\dimI$, we denote the dimension of the intersection of our guess and the true error subspaces. As stated above, if
\begin{equation}
  2(w - \dimI) + \delta  \leq n-k,
  \label{eq:eeCond}
\end{equation}
any Gabidulin error-erasure decoder is able to correct the error, e.g.,~\cite{wachter2013decoding,silva2008rank}.

\begin{lemma}\label{lem:intersectionProb}
  Let $\Uspace$ be a fixed $u$-dimensional $\Fq$-linear subspace of $\Fql$.
  Let $\Vspace$ be chosen uniformly at random from $\Grassm{\Fql}{v}$.
  Then, the probability that the intersection of $\Uspace$ and $\Vspace$ has dimension at least $\omega$ is
  \begin{align*}
\mathrm{Pr}[\dim(\Uspace \cap \Vspace) \geq \omega]    &=\frac{ \sum_{i=\omega}^{\min \{u,v\}} \qbin{\ell-u}{v-i}{q} \qbin{u}{i}{q} q^{(u-i)(v-i)} }{\qbin{\ell}{v}{q}} \\
    &\leq 16 (\min \{u,v\} +1 -\omega)  q^{(\jstar-v)(\ell-u-\jstar)},
  \end{align*}
  where $\jstar := \min\{ v-\omega, \frac{1}{2}(\ell+v-u) \}$.
\end{lemma}

\begin{proof}
See Appendix~\ref{app:proof_lemma}.
\end{proof}

In the following, we analyze guessing only the row space of the error, i.e., $\delta=\gamma$ and $\rho=0$.

\begin{lemma}\label{lem:pr_exact}
    Let $\r'= \m \GGab + \e'  \in \Fqm^n$, where $\rank(\e') = j$, $\e' = \a' \B'$ with $\a' \in \Fqm^{j}$, $\B' \in \Fq^{j\times n}$ and neither parts of the error row space nor column space are known, i.e., $\gamma=\rho=0$ and $t = j$. For $\delta \in [2\xi,n-k]$, the probability that an error-erasure decoder using a random $\delta$-dimensional guess of the error row space outputs $\m \GGab$ is
 \begin{align*}
P_{n,k,\delta,j} &:= \frac{\displaystyle\sum_{i=\lceil j- \frac{n-k}{2} + \frac{\delta}{2}\rceil}^{\min \{\delta,j\}} \qbin{n-j}{\delta-i}{q} \qbin{j}{i}{q}q^{(j-i)(\delta-i)}}{\qbin{n}{\delta}{q}} \\
              &\leq 16 n q^{-(\lceil\frac{\delta}{2}+j- \frac{n-k}{2}\rceil)(\frac{n+k}{2}-\lceil\frac{\delta}{2}\rceil)},
 \end{align*}
 if $2j + \delta > n-k$ and $P_{n,k,\delta,j} := 1$ else.
\end{lemma}

\begin{proof}
First, consider the case where $2j + \delta > n-k$ and define $\xi' := j - \frac{n-k}{2}$.
Let the rows of $\BcHat \in \Fq^{\delta\times n}$ be a basis of the random guess. From~(\ref{eq:eeCond}) follows that if
\begin{align}
  n-k \geq  2 j - 2 \dimI +\delta = n-k + 2\xi' - 2\dimI + \delta, \label{eq:lowbound_col}
\end{align}
where $\dimI$ is the dimension of the intersection of the $\Fq$-row spaces of $\BcHat$ and $\B'$,
an error and erasure decoder is able to decode efficiently. Since $\dimI\leq\delta$, equation~(\ref{eq:lowbound_col}) gives a lower bound on the dimension $\delta$ of the subspace that we have to estimate:
\begin{equation}
2 \xi' \leq 2\dimI-\delta\leq \delta \leq n-k.
\end{equation}
From~(\ref{eq:lowbound_col}) follows further that the estimated space doesn't have to be a subspace of the row space of the error. In fact, it is sufficient that the dimension of the intersection of the estimated column space and the true column space has dimension $\dimI \geq \xi' +\frac{\delta}{2}$.
This condition is equivalent to the condition that the subspace distance (see~\cite{koetter2008coding}) between $\Uspace$ and $\Vspace$ satisfies $d_s(\Uspace,\Vspace):=\dim(\mathcal{U})+\dim(\mathcal{V})-2\dim(\mathcal{U}\cap\mathcal{V})\geq j-2\xi'$. 

From Lemma~\ref{lem:intersectionProb} follows that the probability that the randomly guessed space intersects in enough dimensions such that an error-erasure decoder can decode to one particular codeword in distance $j$ to $\r$ is
  \begin{align*}
    &\frac{\sum_{i=\lceil\xi' + \frac{\delta}{2}\rceil}^{\min \{\delta,j\}} \qbin{n-j}{\delta-i}{q} \qbin{j}{i}{q}q^{(j-i)(\delta-i)}}{\qbin{n}{\delta}{q}} \\
    \leq & 16 \Big(\min\{j,\delta\}+1-\Big(\xi'+\frac{\delta}{2}\Big)\Big)~ q^{-(\lceil\frac{\delta}{2}+\xi'\rceil)(\frac{n+k}{2}- \lceil\frac{\delta}{2}\rceil)} \\
    \leq & 16 n q^{-(\lceil\frac{\delta}{2}+\xi'\rceil)(\frac{n+k}{2}-\lceil\frac{\delta}{2}\rceil)}.
  \end{align*}
For  the case $2j + \delta \leq n-k$, it is well known that that an error erasure decoder always outputs $\m \GGab$.
  \qed
\end{proof}

Lemma~\ref{lem:pr_exact} gives the probability that the error-erasure decoder outputs exactly the codeword $\m \GGab$. Depending on the application, it might not be necessary to find exactly $\m \GGab$ but any codeword $\c \in \CGab$ such that $\rank(\r - \c) \leq w$, which corresponds to Problem~\ref{prob:Search-MLD-Gab}. In the following lemma, we derive an upper bound on the success probability of solving Problem~\ref{prob:Search-MLD-Gab} using the proposed algorithm.

\begin{lemma}\label{lem:pr_any}
 Let $\r$ be a uniformly distributed random element of $\Fqm^n$. Then, for $\delta \in [2\xi,n-k]$ the probability that an error-erasure decoder using a random $\delta$-dimensional guess of the error row space outputs $\c \in \CGab$ such that $\rank(\r - \c) \leq w$ is at most
 \begin{align*}
   \sum_{j=0}^{w} \bar{A}_j P_{n,k,\delta,j} 
\leq 64 n q^{m(k-n)+w(n+m)-w^2 -(\lceil\frac{\delta}{2}+w-\frac{n-k}{2}\rceil)(\frac{n+k}{2}-\lceil\frac{\delta}{2}\rceil) },
 \end{align*}
 where $\bar{A}_j = q^{m(k-n)} \prod_{i=0}^{j-1} \frac{(q^m-q^i)(q^n-q^i)}{q^j - q^i}$.
\end{lemma}
\begin{proof}
Let $\hat{\Code}$ be the set of codewords that have rank distance at most $w$ from the received word, i.e.,
  \begin{align*}
\hat{\Code} := \{ \c \in \CGab : \rank(\r-\c)\leq w  \} = \{\hat{\c}_1,\hdots,\hat{\c}_{\mathcal{N}}\}.
  \end{align*}
  Further, let $X_i$ be the event that the error-erasure decoder outputs $\hat{\c}_i$ for $i = 1,\hdots,\mathcal{N}$ and $\mathcal{A}_j := \{i: \rank(\r-\hat{\c}_i) = j\}$. Observe that $P_{n,k,\delta,j} = \Pr[X_i]$ for $i\in \mathcal{A}_j$, where $\Pr[X_i]$ is the probability that the error-erasure decoder outputs $\hat{\c}_i$ and $P_{n,k,\delta,j}$ is defined as in Lemma 2. Then we can write
  \begin{equation*}
    \Pr[success] = \Pr \Bigg[ \bigcup_{i=1}^{\mathcal{N}} X_i \Bigg] \leq \sum_{i=1}^{\mathcal{N}} \Pr[X_i] = \sum_{j=0}^{w} |\mathcal{A}_j| P_{n,k,\delta,j}.
  \end{equation*}
Let $\bar{A}_j$ be the average cardinality of the set $\mathcal{A}_j$, we have that
  \begin{align*}
    \bar{A}_j = q^{m(k-n)} \prod_{i=0}^{j-1} \frac{(q^m-q^i)(q^n-q^i)}{q^j - q^i} \leq  4 q^{m(k-n)+j(n+m)-j^2}.
    \end{align*}
    Since $\bar{A}_w$ is exponentially larger than $\bar{A}_{w-i}$ for $i>0$, one can approximate
  \begin{align*}
\qquad \qquad ~ ~    \Pr[success] &= \bar{A}_w P_{n,k,\delta,w} \\
               &\leq 64 n q^{m(k-n)+w(n+m)-w^2 -(\lceil\frac{\delta}{2}+w-\frac{n-k}{2}\rceil)(\frac{n+k}{2}-\lceil\frac{\delta}{2}\rceil) }. \qquad     \qed
   \end{align*}

\end{proof}

Based on Lemma~\ref{lem:pr_any}, we can derive a lower bound on the average work factor of Algorithm~\ref{alg:guessRS}.

\begin{theorem}\label{thm:wf}
 Let $\r$ be a uniformly distributed random element of $\Fqm^n$. Then, Algorithm~\ref{alg:guessRS} requires on average at least
 \begin{align*}
\mathcal{W}_{RD} &= \min_{\delta \in [2\xi,n-k]} \left\{\frac{n^2}{ \sum_{j=0}^{w} \bar{A}_j P_{n,k,\delta,j}} \right\} \label{eq:thm1} \\
   &=\!\min_{\delta \in [2\xi,n-k]} \!\left \{\frac{n^2 \qbin{n}{\delta}{q}} {\displaystyle\sum_{j=0}^{\lfloor \frac{n-k-\delta}{2} \rfloor} q^{m(k-n)} \prod_{\ell=0}^{j-1} \frac{(q^m-q^{\ell})(q^n-q^{\ell})}{q^j-q^{\ell}} + \displaystyle\sum_{j=\lfloor \frac{n-k-\delta}{2} \rfloor +1}^{w} \! \!q^{m(k-n)}} \right.  \\
       & \left. \frac{\phantom{1}}{ \dots \Bigg( \displaystyle \prod_{\ell=0}^{j-1} \frac{(q^m-q^{\ell})(q^n-q^{\ell})}{q^j-q^{\ell}}\Bigg) \Bigg(\displaystyle\sum_{i=\lceil j - \frac{n-k}{2} + \frac{\delta}{2} \rceil}^{\min\{\delta,j\}} \qbin{n-j}{\delta-i}{q} \qbin{j}{i}{q} q^{(j-i)(\delta-i)}\Bigg) } \right\}
  \end{align*}
  operations over $\Fqm$ to output $\c \in Gab_k(\g)$, such that $\rank(\r - \c) \le w$, where $\bar{A}_j$ and $P_{n,k,\delta,j}$ are defined as in Lemma~\ref{lem:pr_any}. 
\end{theorem}
\begin{proof}
  Lemma~\ref{lem:pr_any} gives the probability that an error-erasure decoder using a $\delta$ dimensional guess of the row space finds $\c \in \CGab$ such that $\rank(\r - \c) \leq w$.
  This means that one has to estimate on average at least
\begin{equation*}
  \min_{\delta \in [2\xi,n-k]} \left\{\frac{1}{ \sum_{j=0}^{w} \bar{A}_j P_{n,k,\delta,j}} \right\}
\end{equation*}
row spaces in order to output $\c \in Gab_k(\g)$. Since the complexity of error-erasure decoding is in $O(n^2)$, we get a work factor of
\begin{equation*}
\mathcal{W}_{RD} =  \min_{\delta \in [2\xi,n-k]} \left\{\frac{n^2}{ \sum_{j=0}^{w} \bar{A}_j P_{n,k,\delta,j}} \right\}.
\end{equation*}
\qed
\end{proof}

\new{Notice that the upper bound on the probability given in Lemma~\ref{lem:pr_any} is a convex function in $\delta$ and maximized for either $2\xi$ or $n-k$.
Thus, we get the following lower bound on the work factor.
\begin{corollary}
Let $\r$ be a uniformly distributed random element of $\Fqm^n$. Then, Algorithm~\ref{alg:guessRS} requires on average at least
\begin{equation*}
\mathcal{W}_{RD} \geq \frac{n}{64}\cdot q^{m(n-k)-w(n+m)+w^2+\min\{2\xi(\frac{n+k}{2}-\xi),wk\} }
\end{equation*}
operations over $\Fqm$.
\end{corollary}

\begin{remark}
We obtain a rough upper bound of on the expected work factor,
\begin{equation*}
\mathcal{W}_{RD} \leq n^2 q^{m(n-k)-w(n+m)+w^2+\min\{2\xi(\frac{n+k}{2}-\xi),wk\}},
\end{equation*}
by the same arguments as in Lemma~\ref{lem:pr_exact}, Lemma~\ref{lem:pr_any}, and Theorem~\ref{thm:wf}, using
\begin{itemize}
\item lower bounds on the Gaussian binomial coefficient in \cite[Lemma 4]{koetter2008coding},
\item taking the maximal terms in the sums and
\item taking the maximal probability of events instead of union-bound arguments.
\end{itemize}
\end{remark}
}

If $\r \in \Fqm^n$ is defined as in Section~\ref{sec:channel_model}, where neither parts of the error row space nor column space are known, i.e., $\gamma=\rho=0$ and $t = w$, the vector $\r$ can be seen as a uniformly distributed random element of $\Fqm^n$. Thus, Theorem~\ref{thm:wf} gives an estimation of the work factor of the proposed algorithm to solve Problem~\ref{prob:Search-MLD-Gab}. To verify this assumption, we conducted simulations which show that the estimation is very accurate, see Section~5.

\begin{remark}
In Theorem~\ref{thm:wf}, we give a lower bound on the work factor of the proposed algorithm. One observes that especially for small parameters, this bound is not tight which is mainly caused by the approximations of the Gaussian binomials. For larger values, the relative difference to the true work factor becomes smaller.
\end{remark}

Another idea is to guess only the column space or the row and column space jointly.
Guessing the column space is never advantageous over guessing the row space for Gabidulin codes since we always have $n \leq m$. Hence, replacing $n$ by $m$ in the formulas of Lemma~\ref{lem:pr_exact} and in the expression of the probability $P_j$ inside the proof of Theorem~\ref{thm:wf} will only increase the resulting work factor.
For joint guessing, some examples indicate that it is not advantageous, either. See Appendix~\ref{app:joint_guessing} for more details.

\section{Examples and Simulation Results}

We validated the bounds on the work factor of the proposed algorithm in Section~\ref{sec:newAlgo} by simulations.
The simulations were performed with the row/column error-erasure decoder from~\cite{wachter2013decoding} that can correct $t$ rank errors, $\rho$ row erasures and $\gamma$ column erasures up to $2t+\rho+\gamma\leq d-1$. Alternatively, the decoders in~\cite{silva2008rank,GabidulinPilipchuck_ErrorErasureRankCodes_2008} may be considered. One can also observe that the derived lower bounds on the work factor give a good estimate of the actual runtime of the algorithm denoted by $\mathcal{W}_{Sim}$. The results in Table~\ref{tab:sim_results} show further, that for parameters proposed in~\cite{lavauzelle2019, wachter2018repairing}, the new algorithm solves Problem~\ref{prob:Search-MLD-Gab} (\textsf{Search-Gab}) with a significantly lower computational complexity than the approaches based on the known algorithms. 

\new{Therefore, for the RAMESSES system, our algorithm determines the work factor of recovering the private key for all sets of parameters given in \cite{lavauzelle2019}. 
For the modified Faure--Loidreau system, our algorithm provides the most efficient key recovery attack for one set of parameters, shown in Line 5 of Table~\ref{tab:sim_results}.
Notice however that there is a message attack (called \emph{Algebraic Attack} in \cite{wachter2018repairing}) which has smaller complexity.
}

\newcolumntype{x}[1]{>{\centering\arraybackslash\hspace{0pt}}p{#1}}

\begin{table}[ht!]
  \begin{center}
    \caption{Comparison of different work factors for several parameter sets including simulation results for one specific parameter set. \newline 
      $\mathcal{W}_{Sim}$: work factor of the new randomized decoder (simulations). \newline
      $\mathcal{W}_{RD}$: work factor of the new randomized decoder (theoretical lower bound). \newline
      $\mathcal{W}_{Comb}/\mathcal{N}$: work factor of the combinatorial RSD algorithm. \newline
      $\mathcal{W}_{Alg}$: work factor of the algebraic RSD algorithm. \newline
      $\mathcal{W}_{Key}$: work factor of the na\"ive key equation based decoding.
    } \label{tab:sim_results}
	  \def\arraystretch{1.5}
    \begin{tabular}{*{7}{x{0.42cm}}x{1.4cm}x{1.3cm}x{0.9cm}*{4}{x{1.05cm}}}
      \hline\hline
      $q$ & $m$ & $n$ & $k$ & $w$ & $\xi$ & $\delta$ & Iterations & Success & $\mathcal{W}_{Sim}$ & $\mathcal{W}_{RD}$& $\frac{\mathcal{W}_{Comb}}{\mathcal{N}}$ & $\mathcal{W}_{Alg}$ & $\mathcal{W}_{Key}$ 
      \\ \hline\hline\noalign{\vskip 1mm}
      2 & 24 & 24 & 16 & 6 & 2 & 4 & 6844700 & 4488 & $2^{19.74}$ & $2^{19.65}$ & $2^{38.99}$ & $2^{126.01}$ & $2^{43.40}$ \\
	  \hline
      2 & 64 & 64 & 32 & 19 & 3 & 6 & - & - & - & $2^{257.20}$ & $2^{571.21}$ & $2^{460.01}$  & $2^{371.21}$ \\ 
      \hline
      2 & 80 & 80 & 40 & 23 & 3 & 6 & - & - & - & $2^{401.85}$ & $2^{897.93}$ & $2^{576.15}$  & $2^{492.64}$ \\ 
      \hline
      2 & 96 & 96 & 48 & 27 & 3 & 6 & - & - & - & $2^{578.38}$ & $2^{1263.51}$ & $2^{694.93}$ & $2^{589.17}$ \\ 
      \hline
      2 & 82 & 82 & 48 & 20 & 3 & 6 & - & - & - & $2^{290.92}$ & $2^{838.54}$ & $2^{504.70}$ & $2^{410.92}$\\ 
 \hline
    \end{tabular}
  \end{center}
\end{table}

\section{Open Problems}

There is a list decoding algorithm for Gabidulin codes based on Gr\"obner bases that allows to correct errors beyond the unique decoding radius~\cite{horlemann2017module}. 
However, there is no upper bound on the list size and the complexity of the decoding algorithm. 
In future work, the algorithm from~\cite{horlemann2017module} should be adapted to solve Problem~\ref{prob:Search-MLD-Gab} which could allow for estimating the complexity of the resulting algorithm.

%

\bibliographystyle{splncs04}
\bibliography{main}

\begin{appendix}
\section{Proof of Lemma~\ref{lem:intersectionProb}}\label{app:proof_lemma}
  The number of $q$-vector spaces of dimension $v$, which intersections with $\Uspace$ have dimension at least $\omega$, is equal to
  \begin{equation*}
    \sum_{i=\omega}^{\min \{u,v\}}\qbin{\ell-u}{v-i}{q} \qbin{u}{i}{q}q^{(u-i)(v-i)} = \sum_{j=\max\{0,v-u\}}^{v-\omega }\qbin{\ell-u}{j}{q} \qbin{u}{v-j}{q}q^{j(u-v+j)},
  \end{equation*}
  see~\cite{etzion2011error}. Since the total number of $v$-dimensional subspaces of a $\ell$-dimensional space is equal to $\qbin{\ell}{v}{q}$, the probability
\begin{align*}
\mathrm{Pr}[\dim(\Uspace \cap \Vspace) \geq \omega]  &=\frac{ \sum_{i=\omega}^{\min \{u,v\}} \qbin{\ell-u}{v-i}{q} \qbin{u}{i}{q} q^{(u-i)(v-i)} }{\qbin{\ell}{v}{q}} \\
    &= \frac{\sum_{j=\max\{0,v-u\}}^{v-\omega }\qbin{\ell-u}{j}{q} \qbin{u}{v-j}{q}q^{j(u-v+j)}}{\qbin{\ell}{v}{q}}.
  \end{align*}
  Using the upper bound on the Gaussian coefficient derived in~\cite[Lemma 4]{koetter2008coding}, it follows that
  \begin{align*}
\mathrm{Pr}[\dim(\Uspace \cap \Vspace) \geq \omega]
 &\leq
 16\sum_{j=\max\{0,v-u\}}^{v-\omega } q^{j(\ell-u-j)+v(u-v+j)-v(\ell-v)}\\
    &=16\sum_{j=\max\{0,v-u\}}^{v-\omega } q^{(j-v)(\ell-u-j)}\\
    &\leq 16 ~(\min\{u,v\}+1-\omega) q^{(\jstar-v)(\ell-u-\jstar)},
  \end{align*}
  where $\jstar := \min\{ v-\omega, \frac{1}{2}(\ell+v-u) \}$. The latter inequality follows from the fact that the term $(j-v)(\ell-u-j)$ is a concave function in $j$ and is maximum for $j = \frac{1}{2}(\ell+v-u)$.\qed

\section{Guessing Jointly the Column and Row Space of the Error}
\label{app:joint_guessing}

We analyze the success probability of decoding to a specific codeword (i.e., the analog of Lemma~\ref{lem:pr_exact}) for guessing jointly the row and the column space of the error.

\begin{lemma}\label{lem:pr_exact_both}
  Let $\r \in \Fqm^n$ be defined as in Section~\ref{sec:channel_model}, where neither parts of the error row space nor column space are known, i.e., $\gamma=\rho=0$ and $t = w$. The probability that an error-erasure decoder using a random
	\begin{itemize}
\item $\delta_r$-dimensional guess of the error row space and a
\item $\delta_c$-dimensional guess of the error column space,
\end{itemize}
where $\delta_r+\delta_c =: \delta  \in [2\xi,n-k]$, outputs $\m \GGab$ is upper-bounded by
 \begin{equation*}\scriptsize
   \frac{\displaystyle\sum_{i=\lceil\xi + \frac{\delta}{2}\rceil}^{\displaystyle\min \{\delta,w\}} \sum_{\substack{0 \leq w_r,w_c \leq i \\ w_r+w_c=i}} \qbin{n-w}{\delta_r-w_r}{q} \qbin{w}{w_r}{q}q^{(w-w_r)(\delta_r-w_r)} \qbin{m-w}{\delta_c-w_c}{q} \qbin{w}{w_c}{q}q^{(w-w_c)(\delta_c-w_c)}}{\qbin{n}{\delta_r}{q}\qbin{m}{\delta_c}{q}}. \\
 \end{equation*}
\end{lemma}

\begin{proof}
The statement follows by the same arguments as Lemma~\ref{lem:pr_exact}, where we computed the probability that the row space of a random vector space of dimension $\delta$ instersects with the $w$-dimensional row space of the error in $i$ dimensions (where $i$ must be sufficiently large to apply the error erasure decoder successfully). Here, we want that a random guess of $\delta_r$- and $\delta_c$-dimensional vector spaces intersect with the row and column space of the error in exactly $w_r$ and $w_c$ dimensions, respectively. We sum up over all choices of $w_r$ and $w_c$ that sum up to an $i$ that is sufficiently large to successfully apply the error erasure decoder. This is an optimistic argument since guessing correctly $w_r$ dimensions of the row and $w_c$ dimensions of the column space of the error might not reduce the rank of the error by $w_r+w_c$. However, this gives an upper bound on the success probability. \qed
\end{proof}

Example~\ref{ex:joint_guessing} shows that guessing row and column space jointly is not advantageous for some specific parameters.

\begin{example}\label{ex:joint_guessing}
Consider the example $q=2$, $m=n=24$, $k=16$, $w=6$. Guessing only the row space of the error with $\delta = 4$ succeeds with probability $1.66 \cdot 10^{-22}$ and joint guessing with $\delta_r = \delta_c = 2$ succeeds with probability $1.93 \cdot 10^{-22}$. Hence, it is advantageous to guess only the row space (or due to $m=n$ only the column space). For a larger example with $m=n=64$, $k=16$, and $w=19$, the two probabilities are almost the same, $\approx 5.27 \cdot 10^{-82}$ (for $\delta=32$ and $\delta_r=\delta_c=16$).
\end{example}
\end{appendix}

\end{document}